\documentclass[a4paper, oneside, reqno, 12pt]{amsart}
\usepackage[utf8]{inputenc}
\usepackage{amsmath, amssymb, amsfonts}
\usepackage{fullpage}
\usepackage{textcomp, cmap}
\usepackage{mathtools}
\usepackage[nobysame, initials]{amsrefs}
\usepackage[english]{babel}
\usepackage[shellescape,latex]{gmp}

\usepackage{color, xcolor}
\usepackage{graphicx}
\usepackage{pgf}
\usepackage{pgfplots}
\usepackage{relsize}
\usepackage{wrapfig}
\usepackage[justification=centering]{caption}
\usepackage{subcaption}
\theoremstyle{plain}
\newtheorem{theorem}{Theorem}
\newtheorem{conjecture}[theorem]{Conjecture}

\newtheorem{proposition}[theorem]{Proposition}

\newtheorem{lemma}[theorem]{Lemma}
\newtheorem{corollary}[theorem]{Corollary}
\newtheorem{observation}[theorem]{Observation}

\theoremstyle{definition}

\newtheorem{remark}[theorem]{Remark}

\DeclareMathOperator{\conv}{conv}

\DeclareMathOperator{\cone}{cone}
\DeclareMathOperator{\cost}{cost}

\usepackage{mathrsfs}

\usepackage{hyperref}
\hypersetup{
    colorlinks   = true, 
    urlcolor     = blue, 
    linkcolor    = red, 
    citecolor   = green,
    pdftitle= {Proof of the Fingerhut conjecture \today},   
    pdfauthor= {Polina Barabanshchikova and Alexandr Polyanskii}
}
\title[Proof of the Fingerhut conjecture]{Intersecting ellipses induced by a max-sum matching}
\author[P.~Barabanshchikova, A.~Polyanskii]{{Polina Barabanshchikova and Alexandr~Polyanskii}}

\address{Polina Barabanshchikova,
\newline\hphantom{iii} Moscow Institute of Physics and Technology, Institutskiy per. 9, Dolgoprudny, Russia 141700
}
\email{\href{mailto:barabanshchikova.piu@phystech.edu}{barabanshchikova.piu@phystech.edu}}

\address{Alexandr Polyanskii,
\newline\hphantom{iii} Institute of Mathematics and Informatics, Bulgarian Academy of Sciences, Bulgaria, Sofia 1113, Acad. G. Bonchev Str., Bl. 8
\newline\hphantom{iii} Saint Petersburg State University, Line 14 of Vasilyevsky island, 29, St. Petersburg, Russia 199178
}
\email{\href{mailto:alexander.polyanskii@gmail.com}{alexander.polyanskii@gmail.com}}
\urladdr{\url{http://polyanskii.com}}

 
\thanks{
A.P. is supported by the Bulgarian Ministry of Education and Science, Scientific Programme "Enhancing the Research Capacity in Mathematical Sciences (PIKOM)", No. DO1-67/05.05.2022, and by the grant of Russian Science Foundation №22-11-00131, \url{https://rscf.ru/project/22-11-00131/}. Also, A.P. is a Young Russian Mathematics award winner and would like to thank its sponsors and jury.
}
\pgfplotsset{compat=1.17}

\begin{document}

\thispagestyle{empty}

\begin{abstract}
For an even set of points in the plane, choose a \textit{max-sum matching}, that is, a perfect matching maximizing the sum of Euclidean distances of its edges.
For each edge of the max-sum matching, consider the ellipse with foci at the edge's endpoints and eccentricity $\sqrt 3 / 2$.
Using an optimization approach, we prove that the convex sets bounded by these ellipses intersect, answering a Tverberg-type question of Andy Fingerhut from 1995.

\end{abstract}

\maketitle
\section{Introduction}

The celebrated Tverberg theorem~\cite{Tverberg1966} is one of the pearls in combinatorial geometry. It asserts that for any set of $(d+1)(r-1)+1$ points in $\mathbb R^d$, there is a partition with $r$ parts whose convex hulls intersect. One of the ideas for proving it relies on an optimization approach combined with the method of infinite descent. For example, the proofs of the Tverberg theorem by Tverberg and Vre{\'{c}}ica~\cite{Tverberg1993} and by Roudneff~\cite{Roudneff2001} give a good illustration of this idea.

The following variation of the Tverberg theorem was proposed in~\cites{huemer2019matching, soberon2020tverberg} and has been actively studied in recent years. For two points $a,b\in \mathbb R^d$, we denote by $B(ab)$ the closed Euclidean ball for
which the segment $ab$ is its diameter. Let $G$ be a graph whose vertex set is a finite set of points in~$\mathbb R^d$. We say that $G$ is a \textit{Tverberg graph} if
\[
    \bigcap_{xy\in E(G)} B(xy) \neq \emptyset.
\]

One of the problems arising in this context is to find a particular Tverberg graph (perfect matching, Hamiltonian cycle, etc.) for a finite set $S\subset \mathbb R^d$. For that, one can choose a proper function over a particular family of graphs with vertex set $S$ such that a graph maximizing it is a Tverberg graph. For example, it was shown in the papers~\cite{huemer2019matching} and~\cite{Pirahmad2022}*{Theorem 4} that a red-blue perfect matching (for any $n$ red and $n$ blue points in $\mathbb R^d$) maximizing the sum of squared Euclidean distances between the matched points is a Tverberg graph. Also, very recently, the authors of~\cite{Affash2022piercing} proved that a spanning tree (for a finite point set in the plane) maximizing the sum of edge distances is a Tverberg graph.

For an even set $S$ of points in $\mathbb R^d$, let $\mathcal M$ be a perfect matching and put $\cost\;(\mathcal M)\coloneqq \sum_{ab\in \mathcal M} \|a-b\|$. A \textit{max-sum matching} of $S$ is a perfect matching $\mathcal M$ with the maximum value of $\cost(\mathcal M)$. In 2021, Bereg et al.~\cite{bereg2019maximum} showed that for any even set of points in the plane, a \textit{max-sum matching} is a~Tverberg graph. Their motivation was to prove a conjecture of Andy Fingerhut~\cite{Eppstein1995} from~1995.

\begin{conjecture}
    \label{conjecture:fingerhut}
    Let a matching $\{a_ib_i : i=1,\dots,n\}$ be max-sum for a set of $2n$ points in the plane. Then there exists a point $o$ in the plane such that
\begin{equation*}
    \|a_i-o\|+\|b_i-o\|\leq \frac{2}{\sqrt 3} \|a_i-b_i\| \text{\quad for all $i=1,\dots,n$.}
\end{equation*}
\end{conjecture}
Andy Fingerhut~\cite{Eppstein1995} observed that, by the Helly theorem, it is enough to show this conjecture for 6 points, that is, for $n=3$. Also, he mentioned that the factor \( 2/\sqrt 3 \) is the minimum possible. Indeed, consider a six-point set consisting of vertices of two coinciding equilateral triangles. A max-sum matching of this set corresponds to the sides of these triangles. So the center of these equilateral triangles is the only point $o$ satisfying Conjecture~\ref{conjecture:fingerhut}. Hence, one cannot decrease the factor $2/\sqrt{3}$ by a smaller one.

The motivation of Andy Fingerhut was to confirm a conjecture of Subhash Suri~\cite{Suri1998} in designing communication
networks on minimum Steiner stars. For a finite set $S$ of points in the plane, denote by $y$ a point minimizing the sum $\sum_{x\in S} \|y-x\|$. We write $t(S)$ for the corresponding minimum sum. A geometric graph with all edges connecting $y$ with the points of $S$ is called a \textit{minimum Steiner star}.

\begin{conjecture}[S. Suri~\cite{Suri1998}]
    \label{conjecture of suri on minimum tree and max matchin}
    For an even set $S$ of points in the plane with a max-sum matching $\mathcal M$, we have 
\(
    t(S)\leq \frac{2}{\sqrt{3}}\;\cost(\mathcal M) .
\)
\end{conjecture}

In 2000, Fekete and Meijer~\cite{fekete2000minimum} confirmed this conjecture. They derive it from the following observation. (They did not state it as a lemma, however, the next lemma easily follows from their argument.) 
\begin{lemma}[S. P. Fekete and H. Meijer~\cite{fekete2000minimum}]
\label{lemma on matching and angle 2pi/3}
For any even set of points in the plane, there exists a perfect matching \( \mathcal M \) and a point \( o \) in the plane such that either \( o \in \{x,y\} \) or \( \angle xoy \geq 2\pi/3 \) for all \( xy\in \mathcal M\). 
\end{lemma}

To prove this observation, that is, Lemma~\ref{lemma on matching and angle 2pi/3}, Fekete and Meijer apply a Borsuk--Ulam-type argument. Essentially, the same fact with the same proof was also found by Dumitrescu, Pach, and T\'oth in the paper~\cite{dumitrescu2012drawing}*{Lemma 2} from 2012. It is worth mentioning that there are even sets such that no max-matching satisfies Lemma~\ref{lemma on matching and angle 2pi/3}; see the example in \cite{bereg2019maximum}*{Subsection~1.1}. So Conjecture~\ref{conjecture:fingerhut}, a strengthening of Conjecture~\ref{conjecture of suri on minimum tree and max matchin}, is still open.

Andy Fingerhut shared Conjecture~\ref{conjecture:fingerhut} with David Eppstein, who popularized it on his homepage~\cite{Eppstein1995} and proved it with the factor $5/2$ instead of $2/\sqrt{3}$. Bereg et al.~\cite{bereg2019maximum} conclude that their result on max-sum matchings and Tverberg graphs yields the Fingerhut conjecture with the factor $\sqrt 2$. Our goal is to confirm this conjecture completely.

\begin{theorem}
    \label{theorem: fingerhut holds}
    Conjecture~\ref{conjecture:fingerhut} is true.
\end{theorem}

Our proof of Theorem~\ref{theorem: fingerhut holds} combines an optimization approach and the infinite descent method. A similar idea was used in~\cite{Tverberg1993} and~\cite{Roudneff2001} to prove the Tverberg theorem. Also, an analogous approach was applied in~\cite{Pirahmad2022} to find 
a perfect red-blue Tverberg matching of any $n$ red and $n$ blue points in $\mathbb R^d$.

The paper is organized as follows. In Subsection~\ref{subsection ellipses and lenses}, we restate the Fingerhut conjecture in terms of the intersection of convex regions bounded by ellipses. In Subsection~\ref{subsetction auxiliary facts from convex geometry}, we recall classical facts from convex geometry. Finally, we prove the Fingerhut conjecture in Section~\ref{section proof of conjecture}. (One of the technical observations, Lemma~\ref{main_lemma}, is shown in Subsection~\ref{proof of main lemma}.)

\section{Auxiliary notation and facts}

\subsection{Ellipses and lenses}
\label{subsection ellipses and lenses}

For the sake of brevity, we introduce the following notation for the region bounded by an ellipse
\[
    \mathcal E_\lambda(ab)\coloneqq \big\{x\in \mathbb R^2:\|a-x\|+\|b-x\|\leq \lambda \|a-b\|\big\},
\]
where $a$ and $b$ are some points in the plane. Slightly abusing notation, we say that $\mathcal E_\lambda (ab)$ is an ellipse.  Since the factor $\lambda=2/\sqrt{3}$ plays a main role in Conjecture~\ref{conjecture:fingerhut}, we put 
\[
    \mathcal E(ab)\coloneqq\mathcal E_{2/\sqrt{3}}(ab).
\]
So we can easily restate Conjecture~\ref{conjecture:fingerhut} as follows:
\begin{conjecture}
\label{conjecture fingerhut restate}
    If $\mathcal M$ is a max-sum matching of an even set of points in the plane, then
    \[
        \bigcap_{xy\in \mathcal M} \mathcal E(xy)\ne \emptyset.
    \]
\end{conjecture}

Also, we need the notation of lens. For two distinct points $x,y$ in the plane and $\alpha\in (0,\pi)$, denote by $\alpha(xy)$ the so-called \textit{$\alpha$-lens} for the line segment $xy$, that is,
\[
    \alpha(xy)=\{z\in \mathbb R^2\setminus\{x,y\}: \angle xzy\geq \alpha\} \cup \{x,y\}.
\]
Next, we state a trivial observation, which we leave without proof.
\begin{observation}
\label{observation: lens lies in the ellipse}
For any two distinct points $x$ and $y$ in the plane, we have
\[
    \alpha(xy) \subset \mathcal{E}(xy),
\]
where $\alpha=2\pi/3$.
\end{observation}

We refer to~\cite{Pirahmad2022}*{Subsection~9.2} to read about a high-dimensional problem on the intersection of lenses (generalizing Lemma~\ref{lemma on matching and angle 2pi/3} and similar to Conjecture~\ref{conjecture fingerhut restate}).

\subsection{Auxiliary facts from convex geometry}
\label{subsetction auxiliary facts from convex geometry}
We recall a standard notation and theorems from convex geometry; see \cite{HiriartUrruty1993}*{Chapter~6}.

A function $f:\mathbb R^n \rightarrow \mathbb R$ is called \textit{convex} if it satisfies the inequality
\[
    f(\lambda x+(1-\lambda)y)\leq \lambda f(x)+(1-\lambda)f(y)
\]
for any two points $x, y\in \mathbb R^d$ and $0\leq \lambda\leq 1$.
The \textit{subdifferential} $\partial f(x)$ of a convex function~$f$ at a point $x \in \mathbb{R}^n$ is given by
\[
    \partial f(x) \coloneqq \{s \in \mathbb{R}^n : f(y) - f(x) \geq \langle s, y - x \rangle \text{\ \ for all } y \in \mathbb{R}^n \},
\]
where $\langle a, b \rangle$ stands for the dot product of $a,b\in \mathbb R^d$.
According to this definition, the set $\partial f(x)$ is convex. A vector $s \in \partial f(x)$ is called a \textit{subgradient} of $f$ at~$x$.

Next, we state two simple properties of convex functions.
\begin{theorem}{\cite{HiriartUrruty1993}*{Corollary~2.1.4}}\label{theorem: subgradient of differentiable function}
If a convex function $f: \mathbb{R}^n \rightarrow \mathbb{R}$ is differentiable at a point $x \in \mathbb{R}^n$, then its only subgradient at $x$ is its gradient, that is, 
\[
    \partial f(x) = \{\nabla f(x)\}.
\]
\end{theorem}

\begin{theorem}{\cite{HiriartUrruty1993}*{Theorem~2.2.1}}\label{theorem: o lies in partial f(x)}
Let $f: \mathbb{R}^n \rightarrow \mathbb{R}$ be a convex function. Then, $f$ is minimized at $x$ over $\mathbb{R}^n$ if and only if the origin $o$ lies in $\partial f(x)$.
\end{theorem}
We need the following notation to state the next theorem and its corollary.
For convex functions $f_1,\dots,f_m:\mathbb R^n \to \mathbb R$, consider the function $f: \mathbb R^n \to \mathbb R$ defined by $f(x) =  \max \{f_1(x),\dots,f_m(x)\}.$
For any point $x\in \mathbb R$, we write $I(x) =
\{1\leq i\leq m:f_i(x)=f(x)\}.$
\begin{theorem}{\cite{HiriartUrruty1993}*{Corollary~4.3.2}}\label{theorem: subgradient of maximum of convex functions}
For any $x \in \mathbb{R}^n $ we have
$$\partial f(x) = \conv (\cup_{i \in I(x)} \partial f_i (x)).$$
\end{theorem}

The key optimization tool we need is the following corollary of Theorems~\ref{theorem: subgradient of differentiable function}, \ref{theorem: o lies in partial f(x)}, and \ref{theorem: subgradient of maximum of convex functions}.

\begin{proposition}
\label{corollary: key tool}
If the function $f$ attains its global minimum at the point $x$ and, for all $i\in I(x)$, the function $f_i$ is differentiable at $x$, then we have 
\[
    o \in \conv \{\nabla f_i (x) : i \in I(x)\}.
\]
\end{proposition}
\begin{proof}
By Theorem \ref{theorem: subgradient of maximum of convex functions} we have
$$\partial f(x) = \conv (\cup_{i \in I(x)} \partial f_i (x)).$$
Since $f$ attains global minimum at $x$, by Theorem \ref{theorem: o lies in partial f(x)} we get $o \in \partial f(x)$. Applying Theorem~\ref{theorem: subgradient of differentiable function} to the function $f_i$, for all $i\in I(x)$, we obtain
\(
    \partial f_i(x) = \{\nabla f_i(x)\}.
\)
Hence,
\[
o \in \partial f(x) = \conv (\cup_{i \in I(x)} \partial f_i (x)) = \conv \{\nabla f_i (x) : i \in I(x)\},
\]
which finishes the proof.
\end{proof}

In our problem, we apply Proposition~\ref{corollary: key tool} to a particular collection of functions and obtain its corollary in the form needed to prove the Fingerhut conjecture. 

First, we introduce an auxiliary notation. For two vectors $x$ and $y$, put
\[ 
    \ell_{xy}\coloneqq\frac{\lVert y \rVert}{\|x\|+\|y\|}x +\frac{\lVert x\rVert}{\lVert x\rVert+\lVert y\rVert} y.
\]
Clearly, the vector $\ell_{x y}$ bisects the angle between $x$ and $y$, and moreover, the point~$\ell_{xy}$ belongs to the line segment $xy$.

For any distinct points $a$ and $b$, the function $h_{ab}:\mathbb R^2 \to \mathbb R $ is defined by
\[
    h_{a b} (x)\coloneqq \dfrac{\lVert a - x \rVert + \lVert b - x \rVert}{\lVert a - b \rVert}.
\]
Clearly, it is well-defined and differentiable everywhere except at $a$ and $b$. For a set $\mathcal M$ of pairs of distinct points in the plane, the function $h_{\mathcal M}:\mathbb R^2 \to \mathbb R$ is given by $h_{\mathcal M} \coloneqq \max_{ab\in \mathcal M} h_{ab}$.

\begin{corollary} 
\label{corollary on minimum size intersecting ellipses}
Let $\mathcal M$ be a set of pairs of distinct points. If the function $h_{\mathcal M}$ attains its global minimum at the origin $o$ and $h_{\mathcal M}(o) > 1$, then
\[
    o \in \conv \{\ell_{a b} : ab \in \mathcal{M}_1\},
\]
where 
\[
    \mathcal M_1 \coloneqq \left\{ab \in \mathcal{M} : \dfrac{\lVert a \rVert + \lVert b \rVert}{\lVert a - b \rVert} = h_{\mathcal M}(o) \right\}.
\]
\end{corollary}
\begin{proof}
For all $ab\in \mathcal M_1$, we have 
\[
    h_{a b}(o)=h_{\mathcal M}(o)>1=h_{ab}(a)=h_{ab}(b),
\] 
and thus, the points $a$ and $b$ are distinct from $o$. So, for all $ab\in \mathcal M_1$, the function $h_{ab}$ is differentiable at~$o$. 

Applying Proposition~\ref{corollary: key tool} to the function $h(x)=\max_{ab\in \mathcal M} h_{ab}(x)$, we obtain
\[
    o \in \conv \{\nabla h_{a b} (o) : a b \in \mathcal{M}_1 \}.
\]

To finish the proof, we show that the vectors $\nabla h_{ab}(o)$ and $\ell_{ab}$ are collinear in the \textit{opposite} directions. By a routine computation, one easily get that 
\[
\nabla h_{a b} (o) = - \dfrac{1}{\lVert a-b \rVert}\left(\dfrac{a}{\lVert a\rVert}+\dfrac{b}{\lVert b\rVert}\right) =  - \dfrac{\lVert a\rVert + \lVert b\rVert}{\lVert a-b \rVert \lVert a\rVert \lVert b \rVert} \ell_{a b},
\]
which finishes the proof.
\end{proof}

\section{Intersecting ellipses}
\label{section proof of conjecture}
\subsection{Angular notation}
Throughout this section we use the following angular notation. We define the following functions 
\[
    \measuredangle (\cdot, \cdot): \mathbb R^2\setminus\{o\}\to [0, 2\pi) \text{\ \ and\ \ } \angle(\cdot, \cdot):\mathbb R^2\setminus\{o\}\to [0,\pi].
\]
Essentially, $\measuredangle (x, y)$ is the directed angle between vectors $x$ and $y$ measured in counterclockwise direction and $\angle (x, y)$ is the smallest (non-negative) angle between $x$ and $y$. Clearly, these functions are well-defined. Slightly abusing notation, we will also consider the angles $\measuredangle (m,n)$ and $\angle (m,n)$ for rays $m$ and $n$ emanating from the origin $o$, that is, $\measuredangle(m,n)=\measuredangle(x,y)$ and $\angle (m,n)=\angle (x,y)$, where $x\in m$ and $y\in n$ are any two points distinct from $o$.

\subsection{Proof of Theorem~\ref{theorem: fingerhut holds}}
\selectlanguage{english}
Suppose to the contrary that for an even set of points $S\subset \mathbb R^2$ and a max-sum matching $\mathcal{M}$ of $S$, we have
\begin{equation}\label{empty_intersection}
\bigcap\limits_{a b \in \mathcal{M}} \mathcal E(a b) = \varnothing.
\end{equation} 

First, we show that there is no edge with coinciding endpoints, that is, $a\ne b$ for any edge $ab\in \mathcal M$.
Indeed, suppose that $a b \in \mathcal{M}$ and $a=b$. Since the intersection of the ellipses is empty, there is an edge $cd \in \mathcal{M}$ such that the point $a$ lies outside of the line segment $c d$. Hence by the triangle inequality, we have
$$\lVert a - b \rVert + \lVert c - d \rVert = \lVert c - d \rVert  < \lVert a - c \rVert + \lVert a - d \rVert=\lVert a - c \rVert + \lVert b - d \rVert.$$
Thus the cost of the matching $\mathcal M \setminus \{ab,cd\} \cup \{ac,bd\} $ is larger than \(\cost (\mathcal M)\), a contradiction with the maximality of \(\mathcal M \). Therefore, no edge in \( \mathcal M\) has length 0.

Recall that the function $h_{\mathcal M}: \mathbb{R}^2 \xrightarrow{}  \mathbb{R}$ is given by
\[
    h_{\mathcal M}(x)\coloneqq \max\limits_{a b \in \mathcal{M}} \dfrac{\lVert a - x \rVert + \lVert b - x \rVert}{\lVert a - b \rVert},
\]
and so, it is well defined as $a\ne b$ for any $ab\in \mathcal M$. Since $h_{\mathcal M}$ is bounded from below, it attains its minimum at some point $x_{\mathcal{M}}$. Without loss of generality we may assume that $x_{\mathcal{M}}$ coincides with the origin $o$. For the sake of brevity, put 
\[
    \lambda \coloneqq h_{\mathcal M}(o).
\]
By \eqref{empty_intersection} we have $\lambda > 2/\sqrt{3} > 1$. 

Consider the submatching $\mathcal{M}_1 \subseteq \mathcal{M}$ consisting of pairs $ab$ such that 
\[
    \lVert a \rVert + \lVert b \rVert = \lambda \lVert a - b \rVert.
\]
Applying Corollary \ref{corollary on minimum size intersecting ellipses} to the matching $\mathcal M$, we obtain
\[
    o \in \conv \left\{\ell_{ab} : ab \in \mathcal M_1\right\}.
\]
Recall that the vector $\ell_{a b}$ bisects the angle between vectors $a$ and $b$, and the endpoint of  $\ell_{a b}$ lies on the line segment $ab$. 
Let $\mathcal{M}_2$ be a minimal subset of $\mathcal{M}_1$ such that 
\begin{equation}
    \label{equation o lies in the convex hull of l_ab}
    o\in \conv \big\{\ell_{ab}: ab\in \mathcal M_2\big\}.
\end{equation}
By the Carathéodory theorem, the size of $\mathcal M_2$ is 2 or 3. Indeed, $\mathcal M_2$ cannot consist of only one edge $ab$. Otherwise, the origin $o$ lies on the line segment $ab$. Since $o$ also lies on the boundary of the ellipse $\mathcal E_\lambda(ab)$, we conclude that $\lambda=1$, a contradiction.

Denote by $S_0$ the set of all endpoints of line segments of $\mathcal M_2$, which is of size $4$ or $6$. Consider the graph $G$ on the vertex set $V(G)\coloneqq S_0$ with the edge set $E(G)$ partitioned into blue and red subsets
\begin{equation}\label{edge_sets}
E_b(G) \coloneqq \mathcal M_2 \text{\ \ and\ \ } E_r(G) \coloneqq \{xy\in V(G)\times V(G): \lVert x \rVert + \lVert y \rVert < \lambda \lVert x - y \rVert\}.
\end{equation}
By the definition, if $x$ and $y$ are connected by blue edge, then their ellipse $\mathcal E_\lambda(xy)$ contains the origin on its boundary. Therefore, all vertices of $G$ are distinct from the origin. Also, if $xy$ is a red edge, then the origin lies in the interior of $\mathcal E_\lambda (xy)$.

To complete the proof, we show that the graph $G$ contains an \textit{alternating cycle}, that is, a simple cycle of even length whose edges are taken alternately from $E_b(G)$ and $E_r(G)$. Indeed, assume that there is a cycle $x_1 y_1 \dots x_m y_m x_{m+1}$, where $x_{m+1}\coloneqq x_1$, such that $x_i y_i \in E_b(G)$ and $y_ix_{i+1}\in E_r(G)$. Hence we have
\[
    \sum \limits_{i=1}^{m} \lVert y_i - x_{i+1} \rVert > \frac{1}{\lambda} \sum \limits_{i=1}^{m} (\lVert y_i \rVert + \lVert x_{i+1} \rVert) = \frac{1}{\lambda} \sum \limits_{i=1}^{m} (\lVert y_i \rVert + \lVert x_i \rVert) = \sum \limits_{i=1}^{m} \lVert x_i - y_i \rVert.
\]
Replacing in \(\mathcal M\) the blue edges of the cycle by the red ones, we find the perfect matching
\[
\mathcal M\setminus \{x_1y_1,\dots,x_my_m\}\cup \{ y_1x_2,\dots, y_mx_{m+1}\}
\]
with larger cost.
This contradicts the maximality of $\mathcal{M}$.

\begin{remark}
    The authors of~\cite{Pirahmad2022} in their proof of Theorem 3 on the existence of Tverberg matching for even sets in \( \mathbb R^d \) also consider alternating cycles and apply the optimization approach.
\end{remark}

Our proof heavily relies on Observation \ref{observation: lens lies in the ellipse}. 
Namely, we need its corollary.

\begin{corollary}\label{corollary:no-red-edge}
If there is no red edge between vertices $x$ and $y$ of $G$ (that is, they are disjoint or connected by a blue edge), then $\angle (x, y) < 2\pi/3$.
\end{corollary}
\begin{proof}
    By Observation~\ref{observation: lens lies in the ellipse}, the $\alpha$-lens $\alpha(xy)$ with $\alpha=2\pi/3$ lies in the ellipse $\mathcal E(xy)$, which in turn is contained in the interior of $\mathcal E_\lambda(xy)$ as $\lambda>2/\sqrt{3}$. If the vertices $x$ and $y$ are disjoint or connected by a blue edge, then the origin lies out of the interior of $\mathcal E_\lambda(xy)$, and thus, out of $\alpha(xy)$. Hence $\angle (x,y)<2\pi/3$.
\end{proof}

Remark that if no red edge is incident to a vertex $x \in V(G)$, then for any vertex $y \in V(G)\setminus\{x\}$ we have
\[
    \angle (x, y) < 2\pi/3.
\]
Furthermore, the following lemma guarantees the existence of an alternating cycle in this case. Since its proof is routine, we show it in Subsection~\ref{proof of main lemma}.

\begin{lemma}\label{main_lemma}
If there is a vertex $x \in V(G)$ such that $\angle (x, y)<2\pi/3$ for any vertex $y \in V(G)\setminus\{x\}$, then the graph $G$ contains an alternating cycle.
\end{lemma}

\begin{figure}[ht!]
	\centering
	\includegraphics{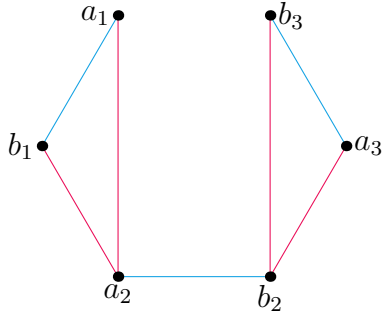}
	\caption{Graph without alternating cycles.}
	\label{fig:im4}
\end{figure}

By Lemma~\ref{main_lemma} and Corollary~\ref{corollary:no-red-edge}, we may assume that each vertex of the graph $G$ is incident to at least one red edge. Recall that $G$ has 4 or 6 vertices and its blue edges induce a perfect matching. Thus one can easily verify that $G$ either contains an alternative cycle or is isomorphic to the graph drawn on Figure~\ref{fig:im4}. Without loss of generality assume that the latter case holds.

\begin{figure}[htp]
    \centering
    \begin{subfigure}[b]{0.49\textwidth}
        \centering
        \includegraphics{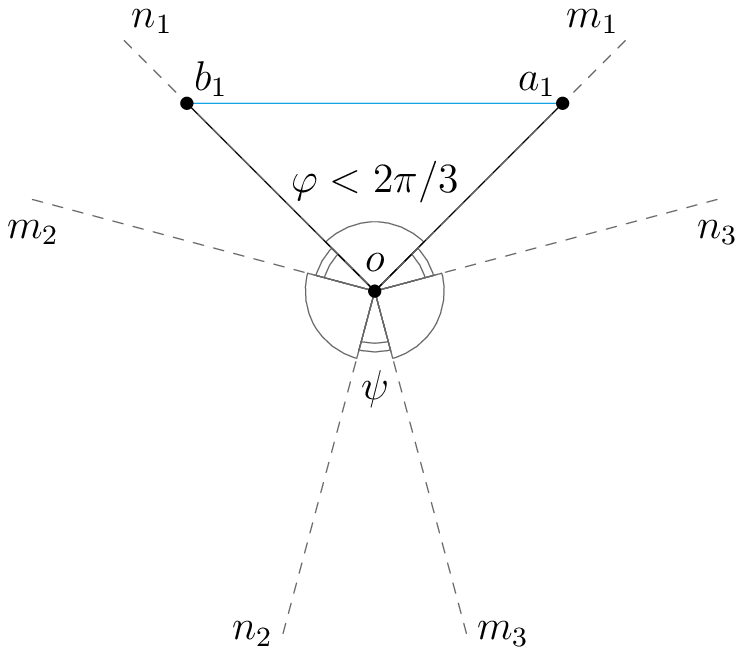}
    \caption{}
    \end{subfigure}
    \begin{subfigure}[b]{0.49\textwidth}
        \centering
        \includegraphics{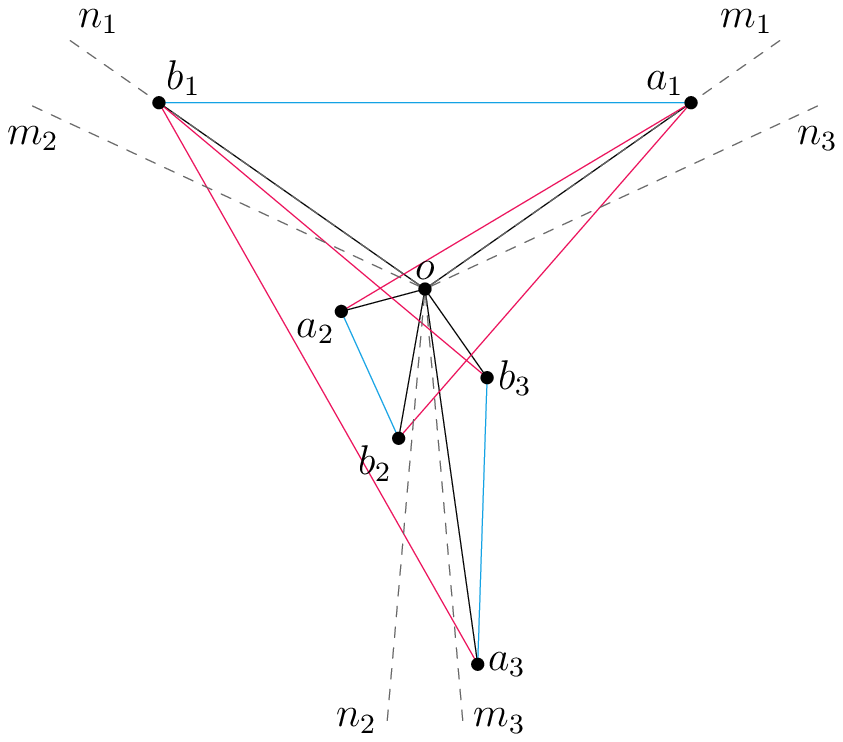}
        \caption{}
    \end{subfigure}
	\caption{Arrangement of the rays and the vertices of $G$.}
	\label{fig:im50}
\end{figure}

Additionally, we put $V(G)=\{a_1,a_2,a_3,b_1,b_2,b_3\}$,
\[
E_b(G) = \{a_1 b_1, a_2 b_2, a_3 b_3\} \text{ and }
E_r(G) = \{a_1 a_2, a_1 b_2, b_1 a_3, b_1 b_3\};
\]
see Figures~\ref{fig:im4} and~\ref{fig:im50}(b). Without loss of generality, we assume that, for all $i$, the directed angle $\measuredangle (a_i, b_i)$ does not exceed $\pi$.

Denote by $m_1$ and $n_1$ the rays emanating from the origin $o$ and passing through the points $a_1$
and $b_1$, respectively; see Figure~\ref{fig:im50}(a). Denote the rays obtained by the counterclockwise rotation of $m_1$ by angles $2\pi / 3$ and $4\pi / 3$ by $m_2$ and $m_3$, respectively. Analogously, define \(n_2\) and \(n_3\); see Figure~\ref{fig:im50}(a). Applying Corollary~\ref{corollary:no-red-edge} to the blue edge \(a_1 b_1\), we have $\angle (b_1, a_1) < 2\pi / 3$. Hence the rays \[
m_1, n_1, m_2, n_2, m_3, \text{ and } n_3
\]
are ordered counterclockwise and divide the plane into 6
regions such that 
\[
    \measuredangle (m_1, n_1)=\measuredangle (m_2, n_2) =\measuredangle (m_3, n_3)=\varphi
\]
and
\[  
    \measuredangle (n_1, m_2) = \measuredangle (n_2, m_3)= \measuredangle (n_3, m_1) =\psi
\]
with $\varphi+\psi=2\pi/3$ and $\psi > 0$;
see Figure~\ref{fig:im50}(a).

Using Corollary \ref{corollary:no-red-edge} for the vertex $a_2$, we have that
\[
    \angle (a_2, y) < 2\pi / 3.
\]
for any vertex $y \in V(G) \setminus \{a_1, a_2\}$.
Therefore, if $\angle (a_1, a_2)<2\pi/3$, we can apply Lemma~\ref{main_lemma} and find an alternating cycle. Hence \( \angle (a_1, a_2) \geq 2\pi / 3\). By Corollary~\ref{corollary:no-red-edge}, we get 
\[
    \angle (b_1, a_2)<2\pi/3 \text{ and thus, }a_2 \in \cone \{m_2, n_2\},
\]
where $\cone A$ stands for the conic hull of a set $A\subseteq \mathbb R^2$.
Analogously, we conclude that
\[
    b_2 \in \cone \{m_2, n_2\},\quad 
    a_3, b_3 \in \cone \{m_3, n_2\}.
\]
  
According to our notation, for each $i\in \{1,2,3\}$ there is the counterclockwise rotation by a non-negative angle that maps the ray $oa_i$ to the ray $ob_i$. 
Hence the rays 
\[
    m_1, n_1, m_2, oa_2, ob_2, n_2, m_3, oa_3, ob_3, n_3
\]
are ordered counterclockwise.

By Corollary~\ref{corollary:no-red-edge}, we have $\angle (a_2,b_3)<2\pi/3$. Since 
\[
    \measuredangle (n_3, m_2)= \measuredangle (n_3, m_1)+\measuredangle (m_1, n_1)+\measuredangle (n_1, m_2)=2\psi+\varphi=2\pi/3+\psi>2\pi/3,
\]
we conclude that $\measuredangle (a_2, b_3)<2\pi/3$, and thus, $b_2, a_3\in \cone\{a_1,b_3\}$; see Figure~\ref{fig:im50}.

The edges $a_2 b_2$ and $a_3 b_3$ belong to $\cone\{m_2, n_2\}$ and $\cone\{m_3, n_3\}$, respectively. Since these cones share only the origin, the line segments $a_2b_2$ and $a_3b_3$ do not intersect. Also, they induce a max-sum matching of the point set $\{a_2,a_3,b_2,b_3\}$. By the triangle inequality, any max-sum matching of four points in a convex position consists of two intersecting line segments. Hence either $b_2$ or $a_3$ lies in the convex hull of the remaining endpoints. Without loss of generality assume that
\[b_2 \in \conv \{a_2, b_3, a_3\}.\]
Since the vectors $a_2$ and $b_2$ are distinct, they are also not collinear.

\begin{figure}[ht!]
    \centering
	\begin{subfigure}[b]{0.45\textwidth}
        \centering
        \includegraphics{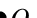}
    \end{subfigure}
    \begin{subfigure}[b]{0.45\textwidth}
        \centering
        \includegraphics{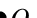}
    \end{subfigure}
	\caption{The point $b_2$ lies inside the triangle $a_2a_3b_3$.}
	\label{fig:imF}
\end{figure}

Denote by $p$ the intersection point of the ray $o b_2$ and the line segments $a_2 b_3$. Analogously, denote by $q$ the intersection point of the ray $ob_2$ and $a_2 a_3$; see Figure~\ref{fig:imF}. So the point $b_2$ lies on the line segment $pq$. Since the ray $ob_2$ does not belong to $\cone\{m_3,n_3\}$ containing $a_3$ and $b_3$, the points $p$ and $q$ are interior for the line segments $a_2b_3$ and $a_2a_3$, respectively.
\medskip

Finally, we finish the proof using the function $f: \{(x,y)\in \mathbb{R}^2 \times \mathbb{R}^2: x\ne y \} \xrightarrow{}  \mathbb{R}$ defined by
\[f(x, y) \coloneqq \frac{\lVert x \rVert + \lVert y \rVert}{\lVert x - y \rVert}.\]
Comparing the definition of $f$ with~\eqref{edge_sets}, we get
\[
    f(a_2, b_3) \geq \lambda,\quad 
    f(a_2, a_3) \geq \lambda,\text{ and }
    f(a_2, b_2) = \lambda.
\]

Recall that for non-collinear vectors $x,y\ne o$ the ray $o\ell_{x y}$ bisects $\cone \{x, y\}$, that is, $\angle (x, \ell_{xy})=\angle (\ell_{xy},y)$. Also, the point $\ell_{x y}$ lies on the line segment $x y$.
\begin{lemma}\label{F_lemma}
Let $x$ and $y$ be non-zero vectors in the plane. Then
\[
    f(x, y) = \dfrac{\lVert x \rVert}{\lVert x - \ell_{x y}\rVert}.
\]
\end{lemma}
\begin{proof}
By the angle bisector theorem, we have
\[
    \dfrac{\lVert x \rVert}{\lVert x - \ell_{x y}\rVert} = \dfrac{\lVert y \rVert}{\lVert \ell_{x y} - y \rVert}.
\]
Hence
\[
    f(x, y) = \frac{\lVert x \rVert + \lVert y \rVert}{\lVert x - y \rVert} = \frac{\lVert x \rVert + \lVert y \rVert}{\lVert x - \ell_{x y} \rVert + \lVert \ell_{x y} - y \rVert} = \dfrac{\lVert x \rVert}{\lVert x - \ell_{x y}\rVert},
\]
which finishes the proof.
\end{proof}

\begin{corollary}\label{F_cor}
Let $x$ and $y$ be non-zero and non-collinear vectors in the plane. For any interior point $z$ of the line segment $x y$, we have
\[
    f(x, z) > f(x, y).
\]
\end{corollary}
\begin{proof}
Since $z$ is an interior point of the line segment $xy$, we have $\angle (x,\ell_{xy})> \angle (x,\ell_{xz})$, and thus,
\[
    \lVert x - \ell_{x y}\rVert > \lVert x - \ell_{x z}\rVert.
\]

By Lemma \ref{F_lemma}, we conclude
\[
    f(x, y) = \dfrac{\lVert x \rVert}{\lVert x - \ell_{x y}\rVert} < \dfrac{\lVert x \rVert}{\lVert x - \ell_{x z}\rVert} = f(x, z),
\]
which finishes the proof.
\end{proof}
Since $p$ and $q$ are interior points of the line segments $a_2 b_3$ and $a_2 a_3$, respectively, Corollary~\ref{F_cor} yields that
\[
    f(a_2, p) > f(a_2, b_3) \geq \lambda \text{\ \ and\ \ } f(a_2, q) > f(a_2, a_3) \geq \lambda.
\]
Next, we show that these inequalities contradict the equality $f(a_2,b_2)=\lambda$.

Since the points $p$ and $q$ lie on the ray $ob_2$, we conclude that the points $\ell_{a_2 p}$ and $\ell_{a_2q}$ lie on the ray $o\ell_{a_2 b_2}$. Moreover, the point $\ell_{a_2b_2}$ belongs to the line segment $\ell_{a_2p}\ell_{a_2q}$, and thus,
\[
    \lVert a_2 - \ell_{a_2 b_2}\rVert \leq \max \big\{\lVert a_2 - \ell_{a_2 p}\rVert,  \lVert a_2 - \ell_{a_2 q}\rVert \big \}.
\]
By Lemma \ref{F_lemma}, we have
\begin{align*}
    \lambda=f(a_2, b_2) = \dfrac{\lVert a_2 \rVert}{\lVert a_2 - \ell_{a_2 b_2}\rVert}  &\geq \min \left\{\dfrac{\lVert a_2 \rVert}{\lVert a_2 - \ell_{a_2 p}\rVert} ,  \dfrac{\lVert a_2 \rVert}{\lVert a_2 - \ell_{a_2 q}\rVert}  \right\} \\
    &=  \min \big\{f(a_2, p),  f(a_2, q) \big\} > \lambda,
\end{align*}
a contradiction. Therefore, the graph $G$ contains an alternating cycle, which finishes the proof of Theorem~\ref{theorem: fingerhut holds}. \hfill \(\square\)
\subsection{Proof of Lemma \ref{main_lemma}}
\label{proof of main lemma}
\begin{figure}[ht!]
    \centering
    \includegraphics{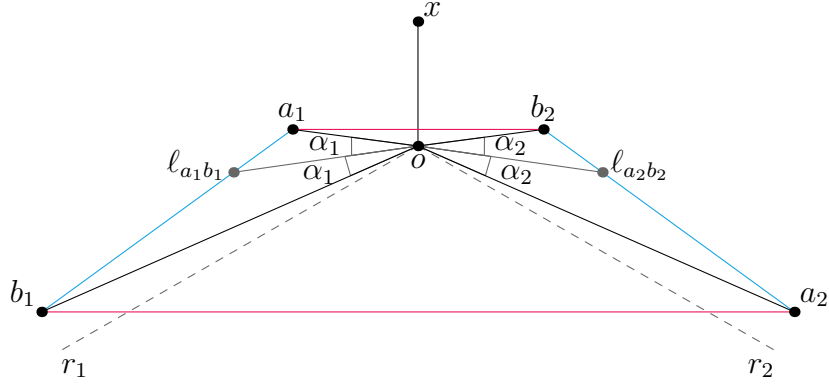}
    \caption{Arrangement of points in the proof of Lemma \ref{main_lemma}.}
    \label{fig:im3}
\end{figure}

Denote by $r_1$ and $r_2$ the rays emanating from the origin such that $\measuredangle (ox, r_1)=\measuredangle (r_2, ox)=2\pi/3$. Thus, we have $\measuredangle (r_1, r_2)=2\pi/3$ as well. Since $\angle (x,y)<2\pi/3$ for any $y\in V(G)\setminus \{x\}$, we conclude that no vertex of $G$ lies in $\cone\{r_1,r_2\}$. 

Let us show that $\ell_{ab}\not\in \cone\{r_1,r_2\}$ for any blue edge $ab\in \mathcal M_2$. By Corollary~\ref{corollary:no-red-edge}, we have $\angle (a,b)<2\pi/3$. Thus, $\cone\{r_1,r_2\}$ and $\cone \{a,b\}$ share only the origin $o$ as $\angle (r_1, r_2)=2\pi/3>\angle (a,b)$. By the definition, the point $\ell_{ab}$ lies on a line segment $ab$, and so, $\ell_{ab}\in \cone \{a,b\}$. Therefore, the point $\ell_{ab}$ does not lie in the $\cone \{r_1,r_2\}$, and thus, $\angle (x,\ell_{ab})<2\pi/3$.

By~(\ref{equation o lies in the convex hull of l_ab}), any closed half-plane bounded by a line through the origin must contain $\ell_{ab}$ for some $ab\in \mathcal M_2$. Among the closed half-planes bounded by the line containing the ray $r_2$, choose the one without $x$. This half-plane must contains $\ell_{ab}$ for some $ab\in \mathcal M_2$. Since no point $\ell_{ab}$ lies in $\cone\{r_1, r_2\}$, we conclude that there is a blue edge $a_1b_1$ such that
\begin{equation}
    \label{equation:a1b1 bounds for oriented angles}
        \pi/3 \leq \measuredangle (x,\ell_{a_1b_1})<2\pi/3;
\end{equation}
see Figure~\ref{fig:im3}.
Analogously, we can find an edge $a_2b_2\in \mathcal M_2$ satisfying the inequality
\begin{equation}
    \label{equation:a2b2 bounds for oriented angles}
    \pi/3\leq \measuredangle (\ell_{a_2b_2},x)< 2\pi/3.
\end{equation}

Moreover, by~(\ref{equation o lies in the convex hull of l_ab}) and the above argument on closed half-planes, we can choose $a_1b_1$ and $a_2b_2$ in such a way that
\begin{equation}
\label{equation:angle between a1b1 and a2b2 less than pi}
    \measuredangle (\ell_{a_1b_1},\ell_{a_2b_2})\leq \pi.
\end{equation}
For that, it is enough to choose the edges $a_1b_1$ and $a_2b_2$ satisfying (\ref{equation:a1b1 bounds for oriented angles}) and (\ref{equation:a2b2 bounds for oriented angles}), respectively, with the maximum possible values of $\measuredangle(x,\ell_{a_2b_2})$ and $\measuredangle (\ell_{a_1b_1},x)$. (Here, we use that no point $\ell_{ab}$ lies in $\cone\{r_1,r_2\}$.)

Without loss of generality assume that for $i\in \{1,2\}$, we have
\[
    \measuredangle (a_i, b_i) = 2\alpha_i,
\]
where $0\leq 2\alpha_i<\pi$. Moreover, Corollary~\ref{corollary:no-red-edge} yields that $\alpha_i<\pi/3$. Hence for $i\in \{1,2\}$, we obtain
\[
    \measuredangle (a_i, \ell_{a_ib_i})=\measuredangle (\ell_{a_ib_i},b_i)=\alpha_i<\pi/3.
\]
Since $b_1\not\in \cone \{r_1, r_2\}$, we have that $b_1\in \cone\{\ell_{a_1b_1},r_1\}$. Analogously, we conclude that $a_2\in \cone \{r_2, \ell_{a_2b_2}\}$. Therefore,
\[
    2\pi/3=\measuredangle (r_1, r_2) < \measuredangle (b_1,a_2) \leq \measuredangle (\ell_{a_1b_1}, \ell_{a_2b_2})\leq \pi,
\]
and so, by Corollary~\ref{corollary:no-red-edge}, the edge $b_1a_2$ is red.
\medskip

To finish the proof, we show that $a_1b_2$ is also red edge.

First, we claim that $a_1\in\cone \{\ell_{a_1b_1},x\}$ and it does not lie on the ray $ox$. Indeed, it follows from 
\[
\measuredangle (a_1, \ell_{a_1,b_1})=\alpha_1<\pi/3\leq \measuredangle(x,\ell_{a_1b_1}).
\]
Analogously, we get that $b_2\in \cone \{x, \ell_{a_2,b_2}\}$ and it does not lie on the ray $ox$. Thus the vertex $x$ is distinct from $b_1$ and $a_2$.

Second, by~(\ref{equation:angle between a1b1 and a2b2 less than pi}), we have
\begin{equation}
    \label{equation: sum of angle less pi }
    \alpha_1+\alpha_2+\angle(b_1,a_2)=\measuredangle (\ell_{a_1b_1},a_1)+\measuredangle (b_1,a_2) + \measuredangle (a_2, \ell_{a_2b_2})=\measuredangle (\ell_{a_1b_1},\ell_{a_2b_2})\leq \pi.
\end{equation}
Since $\angle (b_1, a_2)> \angle (r_1,r_2)=2\pi/3$, we conclude that
\begin{equation}
\label{equation sum of alpha1 and alpha2 less than pi/3} 
    \alpha_1+\alpha_2<\pi/3
\end{equation}
There are two possible cases.

\texttt{1.} The point $x$ does not lie in $\cone \{a_1,b_2\}$. Then either $\angle (a_1,b_2)=\pi$ or $\cone \{a_1,b_2\}$ contains $\cone \{r_1,r_2\}$. In the latter case, we have $\angle (a_1,b_2)\geq 2\pi/3$. Therefore, by Corollary~\ref{corollary:no-red-edge}, the edge $a_1b_2$ is red.

\texttt{2.} The point $x$ lies in $\cone \{a_1,b_2\}$. So, we get
    \begin{align*}
        \angle(b_2,a_1)&=2\pi - \angle (a_1,b_1)-\angle (b_1, a_2)-\angle (a_2,b_2)\\
        &=2\pi - \alpha_1-\alpha_2 -\big(\alpha_1+\alpha_2+\angle(b_2,a_1)\big)\\
        &>2\pi -\pi/3 -\pi=2\pi/3.
    \end{align*}
    Here the last inequality follows from (\ref{equation: sum of angle less pi }) and (\ref{equation sum of alpha1 and alpha2 less than pi/3}). By Corollary~\ref{corollary:no-red-edge}, the edge $a_1b_2$ is red.

Therefore, in these two cases, we find the alternating cycle $a_1b_1a_2b_2$, which finishes the proof of Lemma~\ref{main_lemma}. \hfill \(\square\).

\subsection*{Conflict of interests} Not applicable.
\subsection*{Data availability} Data sharing is not applicable to this article as no datasets were generated or analysed during the current study.
\subsection*{Code availability} Not applicable. 
\subsection*{Acknowledgments.} We are grateful to Fedor Gerasimov, Olimjoni Pirahmad, and the members of the Laboratory of Combinatorial and Geometric Structures at MIPT for the stimulating and fruitful discussions.

\bibliographystyle{siam}
\bibliography{biblio}

\begin{bibdiv}
\begin{biblist}

\bib{Affash2022piercing}{inproceedings}{
      author={Abu-Affash, A.~Karim},
      author={Carmi, Paz},
      author={Maman, Meytal},
       title={Piercing diametral disks induced by edges of maximum spanning
  trees},
        date={2023},
   booktitle={Walcom: Algorithms and computation},
      editor={Lin, Chun-Cheng},
      editor={Lin, Bertrand M.~T.},
      editor={Liotta, Giuseppe},
   publisher={Springer Nature Switzerland},
     address={Cham},
       pages={71\ndash 77},
        note={DOI: \url{https://doi.org/10.1007/978-3-031-27051-2_7}. Available
  at \url{https://arxiv.org/abs/2209.11260}},
}

\bib{bereg2019maximum}{article}{
      author={Bereg, Sergey},
      author={Chac{\'o}n-Rivera, Oscar~P},
      author={Flores-Pe{\~n}aloza, David},
      author={Huemer, Clemens},
      author={P{\'e}rez-Lantero, Pablo},
      author={Seara, Carlos},
       title={On maximum-sum matchings of points},
        date={2023},
     journal={Journal of Global Optimization},
      volume={85},
       pages={111\ndash 128},
        note={DOI: \url{https://doi.org/10.1007/s10898-022-01199-z}. Available
  at \url{https://arxiv.org/abs/1911.10610}},
}

\bib{dumitrescu2012drawing}{article}{
      author={Dumitrescu, Adrian},
      author={Pach, J{\'a}nos},
      author={T{\'o}th, G{\'e}za},
       title={{Drawing Hamiltonian Cycles with no Large Angles}},
        date={2012},
     journal={The Electronic Journal of Combinatorics},
      volume={19},
      number={2},
       pages={P31},
}

\bib{Eppstein1995}{misc}{
      author={Eppstein, David},
       title={{Geometry Junkyard. Centers of maximum matchings}},
        date={1995},
        note={Available at
  \url{https://www.ics.uci.edu/~eppstein/junkyard/maxmatch.html}},
}

\bib{fekete2000minimum}{article}{
      author={Fekete, Sandor~P.},
      author={Meijer, Henk},
       title={On minimum stars and maximum matchings},
        date={2000},
     journal={Discrete \& Computational Geometry},
      volume={23},
      number={3},
       pages={389\ndash 407},
}

\bib{HiriartUrruty1993}{book}{
      author={Hiriart-Urruty, Jean-Baptiste},
      author={Lemar{\'{e}}chal, Claude},
       title={{Convex Analysis and Minimization Algorithms I}},
   publisher={Springer Berlin Heidelberg},
        date={1993},
         url={https://doi.org/10.1007/978-3-662-02796-7},
}

\bib{huemer2019matching}{article}{
      author={Huemer, Clemens},
      author={P{\'e}rez-Lantero, Pablo},
      author={Seara, Carlos},
      author={Silveira, Rodrigo~I},
       title={Matching points with disks with a common intersection},
        date={2019},
     journal={Discrete Mathematics},
      volume={342},
      number={7},
       pages={1885\ndash 1893},
        note={Available at \url{https://arxiv.org/abs/1902.08427}},
}

\bib{Pirahmad2022}{article}{
      author={Pirahmad, Olimjoni},
      author={Polyanskii, Alexandr},
      author={Vasilevskii, Alexey},
       title={Intersecting diametral balls induced by a geometric graph},
        date={2022},
     journal={Discrete \& Computational Geometry},
        note={DOI: \url{https://doi.org/10.1007/s00454-022-00457-x}. Available
  at \url{https://arxiv.org/abs/2108.09795}},
}

\bib{Roudneff2001}{article}{
      author={Roudneff, Jean-Pierre},
       title={{Partitions of points into simplices with $k$-dimensional
  intersection. Part I: the conic Tverberg's theorem}},
        date={2001-07},
     journal={European Journal of Combinatorics},
      volume={22},
      number={5},
       pages={733\ndash 743},
         url={https://doi.org/10.1006/eujc.2000.0493},
}

\bib{soberon2020tverberg}{article}{
      author={Sober{\'o}n, Pablo},
      author={Tang, Yaqian},
       title={{Tverberg’s theorem, disks, and Hamiltonian cycles}},
        date={2021},
     journal={Annals of Combinatorics},
      volume={25},
      number={4},
       pages={995\ndash 1005},
        note={Available at \url{https://arxiv.org/abs/2011.12218}},
}

\bib{Suri1998}{misc}{
      author={Suri, S.},
       title={{Problem 5. Problem session of the 14th ACM Symposium on
  Computational Geometry}},
        date={1998},
        note={Available at
  \url{https://users.cs.duke.edu/~pankaj/scg98-openprobs/open-probs.html}},
}

\bib{Tverberg1966}{article}{
      author={Tverberg, H.},
       title={{A generalization of Radon's theorem}},
        date={1966},
     journal={Journal of the London Mathematical Society},
      volume={s1-41},
      number={1},
       pages={123\ndash 128},
         url={https://doi.org/10.1112/jlms/s1-41.1.123},
}

\bib{Tverberg1993}{article}{
      author={Tverberg, Helge},
      author={Vre{\'{c}}ica, Sini{\v{s}}a},
       title={{On Generalizations of Radon{\textquotesingle}s Theorem and the
  Ham Sandwich Theorem}},
        date={1993-05},
     journal={European Journal of Combinatorics},
      volume={14},
      number={3},
       pages={259\ndash 264},
         url={https://doi.org/10.1006/eujc.1993.1029},
}

\end{biblist}
\end{bibdiv}

\end{document}